\documentclass[journal]{IEEEtran}

\usepackage{cite, graphicx}
\usepackage{hyperref}
\usepackage[cmex10]{amsmath}
\usepackage{amsthm}
\usepackage{amssymb}
\usepackage{bm}
\usepackage{siunitx}
\ifCLASSOPTIONcompsoc
\usepackage[caption=false,font=normalsize,labelfont=sf,textfont=sf]{subfig}
\else
\usepackage[caption=false,font=footnotesize]{subfig}
\fi
\usepackage{algorithmic}
\usepackage{multirow}
\usepackage{mathrsfs}

\theoremstyle{plain}
\newtheorem{thm}{Theorem}

\newtheorem{prop}[thm]{Proposition}

\theoremstyle{definition}

\theoremstyle{remark}

\author{Liumeng~Wang and Sheng~Zhou,~\IEEEmembership{Member,~IEEE}
\thanks{This work is sponsored in part by the National Science Foundation of China (NSFC) under grant No. 61571265, No. 61461136004, and No. 91638204, and Intel Collaborative Research Institute for Mobile Networking and Computing.

Liumeng Wang and Sheng Zhou are with Tsinghua National Laboratory for Information Science and Technology, Dept. of
Electronic Engineering, Tsinghua University, Beijing 100084, China (Email:
wlm14@mails.tsinghua.edu.cn; sheng.zhou@tsinghua.edu.cn).
}}

\title{On the Fronthaul Statistical Multiplexing Gain}
\begin{document}

\maketitle

\begin{abstract}
Breaking the fronthaul capacity limitations is vital to make cloud radio access network (C-RAN) scalable and practical.
One promising way is aggregating several remote radio units (RRUs) as a cluster to share a fronthaul link, so as to enjoy the statistical multiplexing gain brought by the spatial randomness of the traffic. In this letter, a tractable model is proposed to analyze the fronthaul statistical multiplexing gain. We first derive the user blocking probability caused by the limited fronthaul capacity, including its upper and lower bounds. We then obtain the limits of fronthaul statistical multiplexing gain when the cluster size approaches infinity. Analytical results reveal that the user blocking probability decreases exponentially with the average fronthaul capacity per RRU, and the exponent is proportional to the cluster size. Numerical results further show considerable fronthaul statistical multiplexing gain even at a small to medium cluster size.
\end{abstract}


\section{Introduction}
\label{sec:intro}
Recent evolution of radio access network (RAN) features the baseband processing centralization, which enables efficient cooperative signal processing and can potentially reduce the operation and deployment costs.
In a typical cloud radio access network (C-RAN) \cite{CRANwhitepaper}, baseband processing functions are centralized in baseband units (BBUs) while radio functions are integrated in remote radio units (RRUs). BBUs and RRUs are connected with fronthaul network, over which baseband signals are transported. Despite many advantages, one major design challenge of C-RAN is to meet its high fronthaul capacity demand. One typical interface between RRUs and BBUs is Common Public Radio Interface (CPRI) \cite{CPRI}, which can only support fixed-rate delivery of raw baseband signals. For example, 1Gbps fronthaul rate is required even by a single 20MHz LTE antenna-carrier.

To ease the severe burden on the fronthaul, many solutions have been proposed. On the link level, one direct way is to increase the fronthaul capacity, such as using single fiber bi-direction, wavelength-division multiplexing, and etc \cite{chih2014recent}. The other is to reduce the required data rate on the fronthaul, by means of baseband signal compression \cite{lorca2013lossless}, RRU-BBU functionality splitting \cite{SplittingBS}, radio resource allocation \cite{FHCons}, and etc.
On the network level, packet switching can provide hierarchical and flexible fronthaul networking \cite{redesign}, allowing multiple RRUs to form a cluster and share a fronthaul link, which is adopted by the recent IEEE 1914 working group for next generation fronthaul interface (NGFI) \footnote{http://sites.ieee.org/sagroups-1914/}. Moreover, with appropriate RRU-BBU functional split, data rates in the fronthaul links can be traffic-dependent \cite{SplittingBS}. As a result, the randomness of user traffic can be exploited to enjoy the statistical multiplexing gain, in the hope of reducing the fronthaul capacity demand. Some initial numerical studies have validated the fronthaul statistical multiplexing gain brought by the packeterized fronthaul and functional split \cite{FHgain}. Teletraffic theory and event-driven simulations are adopted in \cite{FH2016Gain} to analyze the fronthaul statistical multiplexing gain in C-RAN.

In this paper, we propose a tractable model to quantitatively analyze the fronthaul statistical multiplexing gain. We derive the probability of user blocking due to the limited fronthaul capacity, and obtain its upper and lower bounds. We then use large-limit analysis to get the closed-form expression of the fronthaul statistical multiplexing gain, which enables quantifying the statistical multiplexing gain under large cluster size. Numerical results further confirm that even a small to medium cluster size can obtain notable fronthaul rate reduction.

\section{System Model}
\label{sec:sysmodel}
\begin{figure}[t]
  \centering
  \includegraphics[width=0.35\textwidth]{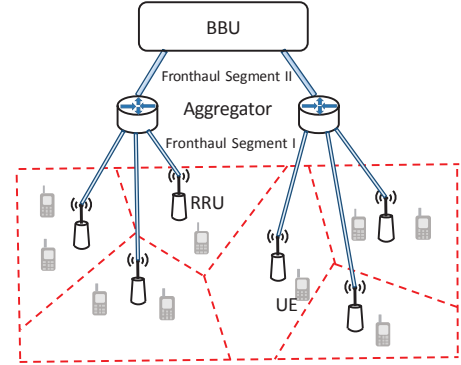}
  \caption{C-RAN architecture with fronthaul links of two-segments with cluster size $K=3$.}
  \label{fig:model}
\end{figure}

Consider a fronthaul network with two segments shown in Fig.\ref{fig:model}.
Fronthaul segment \uppercase\expandafter{\romannumeral1} represents the direct links between RRUs and the aggregators, treated as the "last mile" of the fronthaul network. It can be implemented by CPRI with acceptable costs due to the short transmission distances.
Segment \uppercase\expandafter{\romannumeral2} is the connection between the aggregators and the BBU, and each link is shared by $K$ RRUs formed as a cluster.
Our analysis is focusing on the fronthaul links in segment \uppercase\expandafter{\romannumeral2}.

The capacity of a fronthaul link in segment \uppercase\expandafter{\romannumeral2} is denoted by $T$, shared by an RRU cluster. If the RRU-BBU function is split between layer 1 and layer 2 \cite{SplittingBS}, or with proper fronthaul compression schemes \cite{lorca2013lossless}, the required fronthaul rate depends on the number of users in the corresponding RRUs. We assume that each user requires the same fronthaul rate, and unify the fronthaul capacity $T$ as the maximum number of users whose baseband signals can be transported in the fronthaul link simultaneously.

The spatial distribution of RRUs is modeled by a homogeneous Poisson point process (PPP) with density ${\lambda}_{\text{r}}$, and the spatial distribution of users is also assumed a homogeneous PPP with density ${\lambda}_{\text{u}}$ \cite{Andrew11Tractable}. To simplify the analysis, we assume that there is no coverage overlap between RRUs, i.e., each user is severed by its nearest RRU, and as a result the coverage area of each RRU is a Voronoi cell. We assume that each RRU's location in the cluster is independent to each other\cite{FH2016Gain}, and thus the coverage sizes of different RRUs are also independent. Denote the coverage size of the $i$th RRU as $x_i$, and the distribution of $x_i$ can be approximated by a gamma distribution, i.e., $x_i\thicksim{\Gamma{\left(a,1/({b\lambda_{\text{r}}})\right)}}$,
\begin{equation}
\Gamma{\left(a,1/({b\lambda_{\text{r}}})\right)}={(b\lambda_{\text{r}})}^{a}{\left({x_i}\right)}^{a-1}e^{\left(-b{\lambda}_{\text{r}}{x_i}\right)}/{\Gamma(a)}, \nonumber
\end{equation}
where $a=3.5$, $b=3.5$~\cite{VoronoiCell}, and $\Gamma(\cdot)$ is the Gamma function.
As the scale parameters $\frac{1}{b\lambda_r}$ of all the $K$ gamma distributions are the same, denote the total coverage size of $K$ RRUs as $x=\sum_{i=1}^{K}{x_i}$ , $x\thicksim{\Gamma{\left(Ka,\frac{1}{b\lambda_{\text{r}}}\right)}}$ \cite{GammaSum}. The probability distribution of $x$ given $K$ is then
\begin{equation}\label{eq:clustersize}
f_{x|K}(x)=\frac{{(b\lambda_{\text{r}})}^{Ka}}{{\Gamma(Ka)}}{\left({x}\right)}^{Ka-1}e^{\left(-b{\lambda}_{\text{r}}{x}\right)}.
\end{equation}

Given the RRU cluster coverage size $x$, the total number of users in the RRU cluster , denoted by $N$, obeys the Poisson distribution with mean ${{\lambda}_{\text{u}}}x$, i.e.,
\begin{equation}\label{eq:usernum}
P_K\{N=n|x\} =\frac{({\lambda}_{\text{u}}x)^{n}}{n!}e^{-{{\lambda}_{\text{u}}}x}.
\end{equation}

According to (\ref{eq:clustersize}) and (\ref{eq:usernum}), the distribution of $N$ is
\begin{align} \label{userdist}
P_K\{N=n\}&=\int_0^{\infty}{P_K\{N=n|x\}f_{x|K}(x)}\text{d}x \nonumber \\
&=\int_0^{\infty}\frac{({{\lambda}_{\text{u}}}x)^{n}}{n!}e^{{{\lambda}_{\text{u}}}x}\frac{{(b\lambda_r)}^{Ka}}{{\Gamma(Ka)}}{\left({x}\right)}^{Ka-1}e^{\left(-b{\lambda}_{\text{r}}{x}\right)} \text{d}x \nonumber \\
&=\frac{b^{Ka}{{({\lambda}_{\text{u}}/{\lambda}_{\text{r}})}^n}{\Gamma(n+Ka)}}{{({\lambda}_{\text{u}}/{\lambda}_{\text{r}}+b)}^{Ka+n}{\Gamma(Ka)}n!}.
\end{align}

\section{User Blocking Probability}
\label{sec:blocking}
Due to the fronthaul capacity limitation, the services of the users may be blocked.
If the total number of users per cluster $N$ is larger than the fronthaul capacity $T$, only the baseband signals of $T$ (randomly selected) users will be transmitted in the fronthaul, while services of the remaining $N-T$ users are blocked. The quality of service (QoS) requires the blocking probability be smaller than a threshold  $P_{\text{b}}^{\text{th}}$.
Given the RRU cluster size $K$ and the shared fronthaul capacity $T$, the user blocking probability is expressed as
\begin{align} \label{eq:oribp}
P_{\text{b}}^{K,T}&=\sum_{n=T+1}^{\infty}{P_K\{N=n\}(n-T)/n} \nonumber \\
&=\sum_{n=T+1}^{\infty}{\frac{b^{Ka}(n-T){{({\lambda}_{\text{u}}/{\lambda}_{\text{r}})}^n}{\Gamma(n+Ka)}}{n{({\lambda}_{\text{u}}/{\lambda}_{\text{r}}+b)}^{Ka+n}{\Gamma(Ka)}n!}}.
\end{align}
The expression of $P_{\text{b}}^{K,T}$ is the sum of an infinite series, and it is hard to have intuitive derivations.
Nevertheless, the elements in the summation of (\ref{eq:oribp}) have the following property:
\begin{align} \label{eq:Ratio}
&\frac{P_K\{N=n+1\}(n+1-T)/(n+1)}{P_K\{N=n\}(n-T)/n} \nonumber \\
&=\frac{({\lambda}_{\text{u}}/{\lambda}_{\text{r}})n(n+Ka)(n+1-T)}{({\lambda}_{\text{u}}/{\lambda}_{\text{r}}+b)(n+1)^2(n-T)} \nonumber \\
&=\frac{n+1-T}{n-T}a_n^K\lambda,
\end{align}
where $\lambda=({\lambda}_{\text{u}}/{\lambda}_{\text{r}})/({\lambda}_{\text{u}}/{\lambda}_{\text{r}}+b)$, $a_n^K=n(n+Ka)/(n+1)^2$. As $n>T$, according to the summation in (\ref{eq:oribp}), when $T>2/(Ka-2)$,
\begin{equation} \label{eq:ank}
 1<a_n^K\leq{(T+1)(T+Ka+1)}/{(T+2)^2},
\end{equation}
and only when $n=T+1$, $a_n^K=\frac{(T+1)(T+Ka+1)}{(T+2)^2}$. In fact, as $K \geq 1$ and $T \geq 1$, $T>2/(Ka-2)$ will always hold unless $K=1$ and $T=1$, which is a trivial scenario ignored in our analysis.
We can thus get a lower bound of $P_{\text{b}}^{K,T}$ as
\begin{align} \label{eq:lowerbound}
P_{\text{b}}^{K,T}&>\frac{P_K\{N=T+1\}}{T+1}\sum_{n=T+1}^{\infty}{(n-T)\lambda^{n-T-1}} \nonumber \\
&=\frac{P_K\{N=T+1\}}{(T+1)(1-\lambda)^2} \triangleq P_{\text{b,LB}}^{K,T}.
\end{align}
Similarly, we find an upper bound of $P_{\text{b}}^{K,T}$ as

\begin{align} \label{eq:upperbound}
P_{\text{b}}^{K,T} \nonumber <&\frac{P_K\{N=T+1\}}{T+1} \nonumber \\
&\times \sum_{n=T+1}^{\infty}{(n-T){\left(\frac{(T+1)(T+Ka+1)\lambda}{(T+2)^2}\right)}^{n-T-1}} \nonumber \\
=&\frac{P_K\{N=T+1\}}{T+1}\sum_{n=1}^{\infty}{n{\left(\frac{(T+1)(T+Ka+1)\lambda}{(T+2)^2}\right)}^{n-1}} \nonumber \\
=&\frac{P_K\{N=T+1\}}{(T+1)\left(1-\frac{(T+1)(T+Ka+1)\lambda}{(T+2)^2}\right)^2} \triangleq P_{\text{b,UB}}^{K,T}.
\end{align}
Furthermore, when $T$ approaches infinity,
\begin{equation}
\lim_{T \to +\infty}\frac{P_{\text{b,LB}}^{K,T}}{P_{\text{b,UB}}^{K,T}}=\lim_{T \to +\infty}\frac{\left(1-\frac{(T+1)(T+Ka+1)\lambda}{(T+2)^2}\right)^2}{(1-\lambda)^2}=1,
\end{equation}
which means that when $T$ approaches infinity, the upper bound tends to be equal to the lower bound. As a result, the user blocking probability $P_{\text{b}}^{K,T}$ can be well approximated by either the upper bound or the lower bound when $T$ is large.

\section{Multiplexing Gain Analysis}
\label{sec:gain}
 Denoted by $T_K(P_{\text{b}}^{\text{th}})$ the minimum fronthaul capacity to satisfy a given user blocking probability requirement $P_{\text{b}}^{\text{th}}$ when $K$ RRUs are sharing the fronthaul link, the fronthaul statistical multiplexing gain $G_K(P_{\text{b}}^{\text{th}})$ is accordingly defined as
\begin{equation}
\label{eq:defgain}
G_K(P_{\text{b}}^{\text{th}})=\frac{T_1(P_{\text{b}}^{\text{th}})-T_K(P_{\text{b}}^{\text{th}})/K}{T_1(P_{\text{b}}^{\text{th}})} \nonumber =1-\frac{T_K(P_{\text{b}}^{\text{th}})}{KT_1(P_{\text{b}}^{\text{th}})},
\end{equation}
which represents the relative saving of the equivalent fronthaul capacity per RRU compared with the scenario with no fronthaul sharing, i.e., $K=1$.

The closed-form expression of $G_K(P_{\text{b}}^{\text{th}})$ with general $K$ is hard to derive. Nevertheless, we can still get the closed-form expression with large-$K$ limit, as follows.

\begin{prop} \label{prop:1}
When the cluster size $K$ approaches infinity,
\begin{equation} \label{eq:gaininf}
G_{\infty}(P_{\text{b}}^{\text{th}})=1-\frac{(1-P_{\text{b}}^{\text{th}})\mu}{T_1(P_{\text{b}}^{\text{th}})},
\end{equation}
where $\mu=\lambda_{\text{u}}/\lambda_{\text{r}}$ is the average number of users in the coverage of an RRU.
\end{prop}

\begin{proof}
Given the RRU cluster size $K$, we have $E(N|K)=K\mu$ as the coverage sizes of different RRUs are independent.
According to the law of large numbers, when $K$ approaches infinity, for any positive real $\epsilon$, we have
\begin{equation}\label{eq:largenum}
\lim_{K \to \infty}\text{P}_K(\left|N/K-\mu\right|>\epsilon)=0.
\end{equation}
We define the equivalent fronthaul capacity allocated to each RRU as $\overline{T}=T/K$. If $\overline{T}>\mu$, according to (\ref{eq:largenum}),
\begin{align}
\lim_{K \to +\infty}P_{\text{b}}^{K,T}&=\lim_{K \to +\infty}\sum_{n=K\overline{T}+1}^{\infty}\frac{n-K\overline{T}}{n}P_K(N=n) \nonumber\\
&\leq\lim_{K \to +\infty}\sum_{n=K\overline{T}+1}^{\infty}P_K(N=n) \nonumber\\
&=\lim_{K \to +\infty}P_K(\frac{N}{K}-\mu>\overline{T}-\mu)=0.
\end{align}
If $\overline{T}<\mu$, for any positive real $\epsilon<(\mu-\overline{T})$,
\begin{align}\label{eq:inftyupper}
\lim_{K \to +\infty}P_{\text{b}}^{K,T}&=\lim_{K \to +\infty}\sum_{n=K\overline{T}+1}^{\infty}\frac{n-K\overline{T}}{n}P_K(N=n) \nonumber\\
&=\lim_{K\to+\infty}\sum_{n=K(\mu-\epsilon)}^{n=K(\mu+\epsilon)}\frac{n-K\overline{T}}{n}P_K(N=n) \nonumber \\
&\geq \frac{K(\mu-\epsilon)-K\overline{T}}{K(\mu-\epsilon)}\lim_{K\to+\infty}P_K(\left|N/K-\mu\right|\leq\epsilon) \nonumber \\
&=1-\frac{\overline{T}}{\mu-\epsilon}.
\end{align}
Similarly, we get
\begin{equation} \label{eq:inftylower}
\lim_{K\to+\infty}\sum_{n=K(\mu-\epsilon)}^{n=K(\mu+\epsilon)}\frac{n-K\overline{T}}{n}P_K(N=n)\leq 1-\frac{\overline{T}}{\mu+\epsilon}.
\end{equation}
According to (\ref{eq:inftyupper}) and (\ref{eq:inftylower}),
for any positive real $\delta$, there always exists a positive real $\epsilon<\min\{\frac{\mu^2\delta}{\overline{T}+\mu\delta},\mu-\overline{T}\}$, s.t.
\begin{equation}
\left|\lim_{K\to+\infty}\sum_{n=K(\mu-\epsilon)}^{n=K(\mu+\epsilon)}\frac{n-K\overline{T}}{n}P_K(N=n)-(1-\frac{\overline{T}}{\mu})\right|<\delta,
\end{equation}
i.e., for any positive real $\delta$,
\begin{equation}
\left|\lim_{K \to +\infty}P_{\text{b}}^{K,T}-(1-\frac{\overline{T}}{\mu})\right|<\delta.
\end{equation}
Thus when $K$ approaches infinity,
$\lim_{K \to +\infty}P_{\text{b}}^{K,T}=1-\frac{\overline{T}}{\mu}$,
and the minimum required equivalent fronthaul capacity is
\begin{equation} \label{eq:eqcapacity}
\lim_{K \to +\infty}\frac{T_K(P_{\text{b}}^{\text{th}})}{K}=(1-P_{\text{b}}^{\text{th}})\mu.
\end{equation}
According to (\ref{eq:defgain}) and (\ref{eq:eqcapacity}), we get
\begin{equation}\label{eq:inftypb}
G_{\infty}(P_{\text{b}}^{\text{th}})=\lim_{K \to +\infty}1-\frac{T_K(P_{\text{b}}^{\text{th}})}{KT_1(P_{\text{b}}^{\text{th}})} =1-\frac{(1-P_{\text{b}}^{\text{th}})\mu}{T_1(P_{\text{b}}^{\text{th}})}.
\end{equation}
\end{proof}

This proposition reveals the ultimate multiplexing gain with RRU fronthaul sharing, which is obtained when the RRU cluster size $K$ approaches infinity, and via numerical results we will show that the majority of the gain can be achieved even with small to medium RRU cluster sizes.
Furthermore, we characterize how the RRU cluster size $K$ can expedite the decreasing rate of the user blocking probability $P_b^{K,T}$ w.r.t. the equivalent fronthaul capacity $\overline{T}$ in the following proposition.
\begin{prop}
\label{prop:2}
When $T$ approaches infinity, $P_b^{K,T}$ decreases exponentially with the equivalent fronthaul capacity $\overline{T}$, and the exponent is proportional to $K$.
\end{prop}

\begin{proof}
For fronthaul capacity $T+1$, the user blocking probability can be expressed as
\begin{align}
P_{\text{b}}^{K,T+1}&=\sum_{n=T+2}^{\infty}{\frac{b^{Ka}(n-T-1){{({\lambda}_{\text{u}}/{\lambda}_{\text{r}})}^n}{\Gamma(n+Ka)}}{n{({\lambda}_{\text{u}}/{\lambda}_{\text{r}}+b)}^{Ka+n}{\Gamma(Ka)}n!}} \nonumber \\
&=\sum_{n=T+1}^{\infty}{\frac{b^{Ka}(n-T){{({\lambda}_{\text{u}}/{\lambda}_{\text{r}})}^{n+1}}{\Gamma(n+1+Ka)}}{(n+1){({\lambda}_{\text{u}}/{\lambda}_{\text{r}}+b)}^{Ka+n+1}{\Gamma(Ka)}(n+1)!}} \nonumber \\
&=\sum_{n=T+1}^{\infty}{a_n^K\lambda{\frac{b^{Ka}(n-T){{({\lambda}_{\text{u}}/{\lambda}_{\text{r}})}^n}{\Gamma(n+Ka)}}{n{({\lambda}_{\text{u}}/{\lambda}_{\text{r}}+b)}^{Ka+n}{\Gamma(Ka)}n!}}}.
\end{align}
When $T>2/(Ka-2)$, according to (\ref{eq:ank}), we have
\begin{equation} \label{eq:tul}
\lambda<\frac{P_{\text{b}}^{K,T+1}}{P_b^{K,T}}<\frac{(T+1)(T+Ka+1)}{(T+2)^2}\lambda.
\end{equation}
Taking the limit of (\ref{eq:tul}) with $T$ approaching infinity, leads to
\begin{align} \label{eq:limit}
&\lim_{T \to +\infty}\frac{P_{\text{b}}^{K,T+1}}{P_{\text{b}}^{K,T}}=\lambda, \nonumber \\
&\lim_{T \to +\infty}\log(P_{\text{b}}^{K,T+1})-\log(P_{\text{b}}^{K,T})=\log(\lambda).
\end{align}
Note that $\lambda=({\lambda}_{\text{u}}/{\lambda}_{\text{r}})/({\lambda}_{\text{u}}/{\lambda}_{\text{r}}+b)$.
According to (\ref{eq:limit}),
\begin{equation}
\lim_{T \to +\infty}\frac{\mathrm{d}\log(P_{\text{b}}^{K,T})}{\mathrm{d}{\overline{T}}}=K\log(\lambda).
\end{equation}
\end{proof}

\section{Numerical Results}
\label{sec:num}
The average number of users per RRU $\lambda_{\text{u}}/\lambda_{\text{r}}$ is set to $5$ in our numerical study.
For the simulations, we consider totally 1000 RRUs, and each $K=1, 3, 5, 10, 20, 50$, and 100 RRUs are selected as a cluster, respectively.
The user blocking probability w.r.t. the equivalent fronthaul capacity per RRU $\overline{T}$ is presented in Fig.~\ref{fig:bp}. With the same $\overline{T}$, the larger the cluster size $K$, the smaller the user blocking probability will be. With larger $K$, the user blocking probability also decreases more rapidly. Note that in Fig.~\ref{fig:bp}, the y-axis of the smal rectangle is decimal while the large y-axis is logarithmic.
\begin{figure}[!t]
  \centering
  \includegraphics[width=0.4\textwidth]{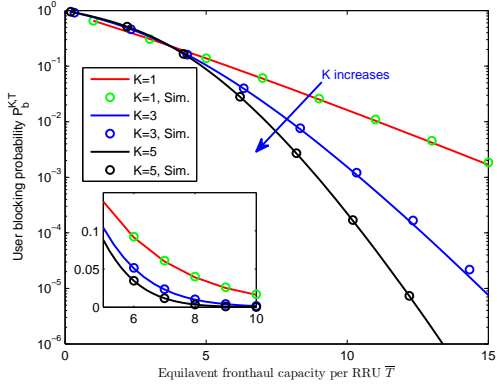}
  \caption{User blocking probability versus equivalent fronthaul capacity per RRU $\overline{T}$ under different cluster size $K$.}
  \label{fig:bp}
\end{figure}

\begin{figure}[!t]
  \centering
  \includegraphics[width=0.4\textwidth]{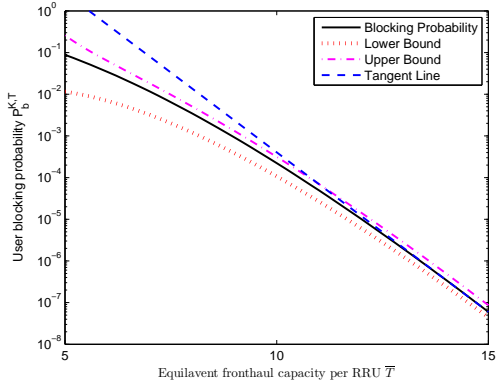}
  \caption{The upper bound and lower bound of the user blocking probability when $K=5$.}
  \label{fig:appr}
\end{figure}
The upper and lower bounds of the user blocking probability are presented in Fig.~\ref{fig:appr} when $K=5$. We can see that with medium to large equivalent fronthaul capacity $\overline{T}$, the blocking probability can be well approximated by the upper bound and the lower bound. The tangent line shows that the user blocking probability decreases exponentially with $\overline{T}$ when $\overline{T}$ is large, which agrees with Proposition \ref{prop:2}.

If the required user blocking probability $P_{\text{b}}^{\text{th}}$ is $5\%$, as shown in Fig.~\ref{fig:gain}, when $K=1$, the minimum required fronthaul capacity $T_1(0.05)=8$, while $T_5(0.05)/5=5.8$, the corresponding fronthaul statistical multiplexing gain $G_5(0.05)=27.5\%$. When cluster size $K$ approaches infinity, according to (\ref{eq:gaininf}), $G_{\infty}(0.05)=40.6\%$.
\begin{figure}[!t]
  \centering
  \includegraphics[width=0.4\textwidth]{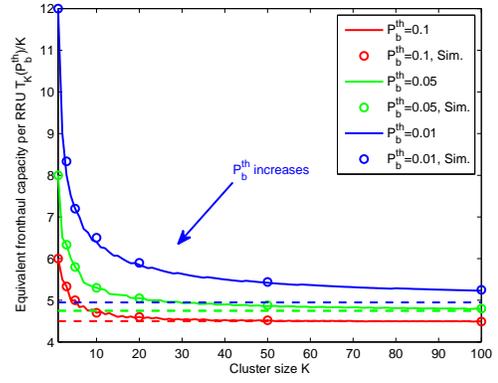}
  \caption{Required equivalent fronthaul capacity $T_K(P_b^{th})/K$ versus RRU cluster size $K$. Dashed-lines denote $T_{\infty}(P_b^{th})$.}
  \label{fig:gain}
\end{figure}
Fig.~\ref{fig:gain} also shows that when the cluster size $K$ increases, the equivalent fronthaul capacity required per RRU tends to decrease. The decreasing rate slows down when $K$ increases, and a small to medium sized cluster can obtain substantial statistical multiplexing gain close to the large-$K$ limit in Proposition \ref{prop:1}.

\section{Conclusions}
\label{sec:future}
In this paper, we have proposed a tractable model to analyze the fronthaul statistical multiplexing gain in C-RAN. We derive the user blocking probability $P_{\text{b}}^{K,T}$ given the RRU cluster size $K$ and the shared fronthaul capacity $T$. We further get the upper bound and the lower bound of $P_{\text{b}}^{K,T}$. We then analyze the statistical multiplexing gain and find that user blocking probability decreases exponentially with the average fronthaul capacity, and the exponent is proportional to the cluster size. Numerical results further reveal that through fronthaul multiplexing, the required average fronthaul capacity per RRU can be notably reduced, and even a small to medium cluster size can obtain considerable fronthaul capacity savings.

\bibliographystyle{IEEEtran}
\bibliography{IEEEabrv,FINAL}

\end{document}